\documentclass[fleqn,10pt]{wlscirep}

\usepackage{amsthm}
\usepackage{amsfonts}
\usepackage{makeidx}
\usepackage{amsmath}
\usepackage{amssymb}
\usepackage{graphicx}
\usepackage{multirow}
\usepackage{verbatim}
\usepackage{url}
\usepackage{subfigure}
\usepackage{braket}
\usepackage{array}

\usepackage{longtable}
\usepackage{rotating}

\usepackage[ruled]{algorithm2e}

\SetAlFnt{\small}
\SetAlCapFnt{\small}
\SetAlCapNameFnt{\small}
\SetAlCapHSkip{0pt}
\IncMargin{-\parindent}

\newtheorem{lem}{Lemma}
\newtheorem{Prop}{Proposition}
\newcommand{\tabincell}[2]{\begin{tabular}{@{}#1@{}}#2\end{tabular}}

\begin{document}

\title{Further Theoretical Study of Distribution Separation Method for Information Retrieval}

\author[1]{Peng Zhang}
\author[1,2]{Qian Yu}
\author[1,*]{Yuexian Hou}
\author[1,3,*]{Dawei Song}
\author[1]{Jingfei Li}
\author[4]{Bin Hu}
\affil[1]{Tianjin University, China}
\affil[2]{The Hong Kong Polytechnic University, Hong Kong}
\affil[3]{The Open University, United Kingdom}
\affil[4]{Lanzhou University, China}

\affil[*]{\{yxhou,dwsong\}@tju.edu.cn}


\begin{abstract}
Recently, a Distribution Separation Method (DSM) is proposed for relevant feedback in information retrieval, which aims to approximate the true relevance distribution by separating a seed irrelevance distribution from the mixture one. While DSM achieved a promising empirical performance, theoretical analysis of DSM is still need further study and comparison with other relative retrieval model. In this article, we first generalize DSM's theoretical property, by proving that its minimum correlation assumption is equivalent to the maximum (original and symmetrized) KL-Divergence assumption. Second, we also analytically show that the EM algorithm in a well-known Mixture Model is essentially a distribution separation process and can be simplified using the linear separation algorithm in DSM. Some empirical results are also presented to support our theoretical analysis.
\end{abstract}

\flushbottom
\maketitle
%
%
\thispagestyle{empty}

\section*{Introduction}
Relevant feedback is an effective method in information retrieval, which can significantly improve the retrieval performance. However, the approximation of relevance model is usually still a mixture model containing an irrelevant component. A Distribution Separation Method is recently proposed to solve this problem\cite{ZhangSigir09}.
The formulation of the basic DSM was based on two assumptions, namely the linear combination assumption and minimum correlation assumption. The former assumes that the mixture term distribution is a linear combination of the relevance and irrelevance distributions, while the later assumes that the relevance distribution should have the minimum correlation with the irrelevance distribution. The basic DSM provided a lower bound analysis for the linear combination coefficient, based on which the desired relevance distribution can be estimated.
It was also proved that the lower bound of the linear combination coefficient corresponds to the condition of the minimum Pearson correlation coefficient between DSM's output relevance distribution and an input seed irrelevance distribution.

In this article, we theoretically extend the generality of the aforementioned linear combination analysis and the minimum correlation analysis of DSM. First, we propose to explore the effect of DSM on the KL-divergence between DSM's output distribution and the seed irrelevance distribution. We    theoretically prove that the lower-bound analysis can also be applied to KL-divergence, and the minimum correlation coefficient corresponds to the maximum KL-divergence. We further prove that the decreasing correlation coefficient leads to a maximum symmetrized KL-divergence as well as JS-divergence.

Second, we investigate the relations between DSM and a well-known Mixture Model Feedback (MMF) approach \cite{Zhai01model-basedfeedback} in information retrieval. We theoretically show that the EM-based iterative algorithm in MMF is essentially a distribution separation process and thus its iterative steps can be simplified by the linear separation technique developed in DSM without decline of performance.



\section*{Basic Analysis of DSM}
\label{sec:DSM}

In this section, we briefly describes assumptions and analysis of the basic DSM \cite{ZhangSigir09}. We use $M$ to represent the mixture term distribution derived from all the feedback documents, and we believe that $M$ is a mixture of relevance term distribution $R$ and irrelevance term distribution $I$. In addition, we assume that only part of the irrelevance distribution $I_S$ (also called as seed irrelevance distribution) is available, while the other part of irrelevance distribution is unknown (denoted as $I_{\overline{S}}$).

The task of DSM can be defined as follows: given the mixture
distribution $M$ and a seed irrelevance distribution $I_{S}$, to derive an
output distribution that can approximates the $R$ as closely as possible. Specifically, as shown in Figure~\ref{fig:linearCombi}, the task of DSM can be divided into two problems:
(1) How to separate  $I_S$ from $M$, and derive $l(R,I_{\overline{S}})$, which
is less noisy but is still a mixture of the true relevance distribution ($R$) and the
unknown irrelevance distribution ($I_{\overline{S}}$).
(2) How to further refine the derived distribution $l(R,I_{\overline{S}})$ to approximate $R$ as closely as possible?
\begin{table}[t]
\centering
\begin{tabular}{c|l}
\textbf{Notation} & \textbf{Description}\\ \hline
$M$ & Mixture term distribution \\
$R$ & Relevance term distribution \\
$I$ & Irrelevance term distribution.\\
$I_S$ & Seed Irrelevance  distribution \\
$I_{\overline{S}}$ & Unknown Irrelevance distribution\\
$F(i)$ & Probability of the $i^{th}$ term in any distribution $F$ \\
$l(F,G)$ & Linear combination of distributions $F$ and $G$ \\
\end{tabular}
\caption{Notations\label{tab:Notation}}
\end{table}
\begin{figure}[t]
\centering
\includegraphics[width=2in]{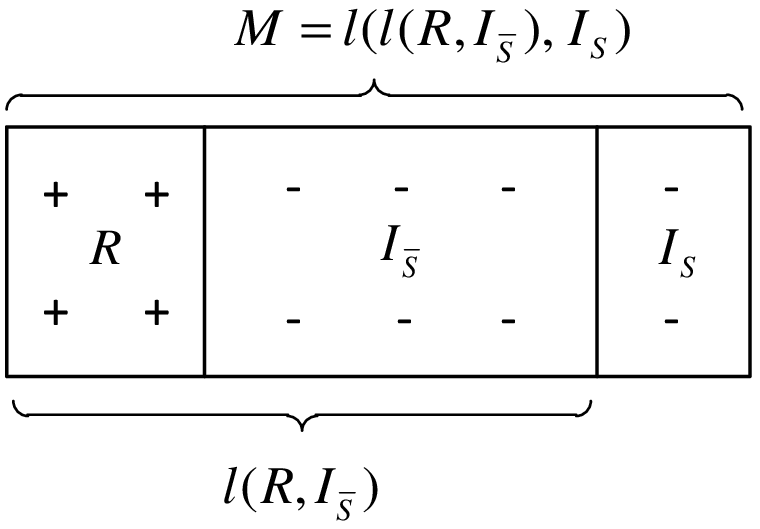}
\caption{An illustration of the linear combination $l(\cdot,\cdot)$
between two term distributions. } \label{fig:linearCombi}
\end{figure}
To solve the above two problems, DSM assumes that a term distribution derived from a feedback document set $D$ is a linear combination of two term distributions, which are derived from two document subsets that form a partition of $D$. Under such a condition, the mixture distribution $M$ derived from all the feedback documents can be a linear combination of $R$ (derived from relevant documents) and  $I$ (derived from irrelevant documents). As shown in Figure~\ref{fig:linearCombi}, $M$ can be a linear combination of two distributions $I_{S}$ and
$l(R,I_{\overline{S}})$ , where
$l(R,I_{\overline{S}})$ is also a linear combination of $R$ and $I_{\overline{S}}$.
Bear in mind that in the linear combination of $R$ and $I$,  both $R$ and $I$ are unknown for DSM. Therefore, determining the linear combination of $I_{S}$ and $l(R,I_{\overline{S}})$ and the linear combination of  $R$ and $I_{\overline{S}}$ are key for solving the above problem (1) and (2), respectively.


%

\subsection*{Linear Combination Analysis}\label{subsec:DMRSND}


Since $M$ is a nested linear combination
$l(l(R,I_{\overline{S}}),I_{S})$, it can be represented
as:
\begin{equation}\label{eqa:firstLC}
M = \lambda \times l(R,I_{\overline{S}}) + (1-\lambda)\times I_{S}
\end{equation} 
where $\lambda$ $(0 \!< \!\lambda \leq 1)$ is the real linear coefficient.
The problem of estimating $l(R,I_{\overline{S}})$ does not have a unique solution generally with coefficient $\lambda$ is unknown.
Therefore, the key is to estimate $\lambda$. Let $\hat{\lambda} (0 <\! \hat{\lambda}\! \leq 1)$ denote an
estimate of $\lambda$, and correspondingly let
$\hat{l}(R,I_{\overline{S}})$ be the estimation of the desired
distribution $l(R,I_{\overline{S}})$. According to Eq.~\ref{eqa:firstLC}, we have
\begin{equation}\label{eqa:computeR}
\hat{l}(R,I_{\overline{S}}) = \frac{1}{\hat{\lambda}}\times M +
(1-\frac{1}{\hat{\lambda}})\times I_S.
\end{equation}
Then with the constraint that all the values in the distribution
$\hat{l}(R,I_{\overline{S}})$ are not less than 0, we can have
\begin{equation}
\label{eqa:inequationlambda} \hat \lambda \times {\bf
1} \succcurlyeq ({\bf
1}-M./I_{S})
\end{equation}
where ${\bf 1}$ stands for a vector in which all the entries are 1, and
$./$ denotes the entry-by-entry division of $M$ by $I_{S}$. Note
that if there is zero value in $I_{S}$, then $\hat \lambda
> 1-\infty$. It is still valid since $\hat \lambda >
0$.
Effectively, Eq.~\ref{eqa:inequationlambda} sets a lower bound of
$\hat{\lambda}$ :
\begin{equation}\label{eqa:l_lambda}
\lambda_{L} = \max{({\bf 1}-M./I_{S})}
\end{equation}
where $\max({\cdot})$ denotes the max value in the resultant vector
${\bf 1}-M./I_{S}$.
The lower bound $\lambda_L$ itself also
determines an estimate of $l(R,I_{\overline{S}})$, denoted as
$l_{L}(R,I_{\overline{S}})$.


\begin{table}[t]
\centering
{
\begin{tabular}{l|l|l}
\textbf{Original} & \textbf{Simplified} & \textbf{Linear Coefficient}\\
\hline $l(R,I_{\overline{S}})(i)$ & $l(i)$ & $\lambda$ \\
 $\hat{l}(R,I_{\overline{S}})(i)$ & $\hat{l}(i)$ & $\hat{\lambda}$ (estimate of $\lambda$)\\
 $l_L(R,I_{\overline{S}})(i)$ & $l_L(i)$ & $\lambda_{L}$ (lower bound of $\hat{\lambda}$)\\
\end{tabular}}
\caption{Simplified Notations\label{tab:SimpleNotation}}
\end{table}

The lower bound $\lambda_{L}$ is essential to the estimation of
${\lambda}$. Now, we present an important property of $\lambda_L$ in
Lemma~\ref{lemma:zeroValue}. For simplicity, we use some
simplified notations listed in Table~\ref{tab:SimpleNotation}.
Lemma~\ref{lemma:zeroValue} guarantees that if there exists zero
value (e.g., for a term $i$, $l(i)=0$) in $l(R,I_{\overline{S}})$, then $\lambda = \lambda_{L}$. In relevance feedback, if there is no distribution smoothing step involved for feedback model, zero values often exist in $l(R; I_S)$, leading to the distribution ${l}_{L}(R,I_{\overline{S}})$ $w.r.t.$ $\lambda_L$ being exactly the desired
distribution $l(R,I_{\overline{S}})$ $w.r.t.$ $\lambda$.

\begin{lem}
\label{lemma:zeroValue}
If there exists a zero value in $l(R,I_{\overline{S}})$, then
$\lambda = \lambda_{L}$, leading to $l(R,I_{\overline{S}}) =
l_{L}(R,I_{\overline{S}})$.
\end{lem}

The proof can be found in~\cite{ZhangSigir09}.
In the IR background, after applying smoothing (usually with the collection
model), there would be no zero values, but instead a lot of small
values exist, in $l(R,I_{\overline{S}})$. In this case,
Remark~\ref{remark:approLambda} guarantees the approximate equality
between $\lambda_L$ and $\lambda$. The detailed description of this remark can be found in~\cite{ZhangSigir09}.

\newtheorem{remark}{Remark}
 \begin{remark}\label{remark:approLambda}
 If there is no zero value, but there exist a few very small values
 in $l(R,I_{\overline{S}})$, i.e., $0<{l}(i)\leq \delta$, where
 $\delta$ is a very small value, then $l_{L}(R,I_{\overline{S}})$
 will be approximately equal to $l(R,I_{\overline{S}})$.
 \end{remark}

\subsection*{Minimum Correlation Analysis}
\label{subsec:MCA}

In this section, we go in-depth to see another property of the combination coefficient and its lower bound. Specifically,
we analyze the correlation between $\hat{l}(R,I_{\overline{S}})$ and $I_S$, along with the decreasing coefficient $\hat \lambda$. Pearson product-moment correlation
coefficient \cite{citeulike:361042}, denoted as $\rho$ $(-1\leq \rho \leq 1)$,
is used as the correlation measurement.

\begin{Prop}\label{pro:minimumcorrelation}
If $\hat{\lambda}$ $(\hat{\lambda}>0)$ decreases, the correlation
coefficient between $\hat{l}(R,I_{\overline{S}})$ and $I_S$, i.e.,
$\rho(\hat{l}(R,I_{\overline{S}}),I_S)$, will decrease.
\end{Prop}

More precisely, the condition of $\hat{\lambda}$ is $\lambda_L  \leq \hat{\lambda} \leq 1$ to ensure that the output distribution $\hat{l}(R,I_{\overline{S}})$ is a distribution. The proof of Proposition~\ref{pro:minimumcorrelation} can be found in~\cite{ZhangSigir09}.
According to Proposition~\ref{pro:minimumcorrelation}, among all
$\hat{\lambda} \in [\lambda_L,1]$, $\lambda_L$ corresponds to
$\min(\rho)$, i.e., the minimum correlation coefficient between
$\hat{l}(R,I_{\overline{S}})$ and $I_S$.

We can also change the \textit{minimum correlation coefficient} (i.e., $\min{(\rho)}$) to \textit{minimum squared correlation coefficient} (i.e.,
$\min{(\rho^2)}$). This idea
can be formulated as the following optimization problem:
\begin{equation}\label{eqa:minimumCorrelation2}
\begin{split}
&\min_{\hat{\lambda}}{[\rho(\hat{l}(R,I_{\overline{S}}),\ I_S)]^2}\ \ \ \  s.t. \quad \lambda_L \leq \hat{\lambda} \leq 1
\end{split}
\end{equation}

To solve this optimization problem, we need to first find a
$\hat{\lambda}$ such that the corresponding
$\rho(\hat{l}(R,I_{\overline{S}}),\ I_S)=0$. According to the proof
of Proposition~\ref{pro:minimumcorrelation} in~\cite{ZhangSigir09}, this
$\hat{\lambda}=-\frac{a}{b}$, where $a=\sum_{i}^{m}({I_S}(i) -
\frac{1}{m})({M}(i)-{I_S}(i))$, $b=\sum_{i}^{m}({I_S}(i) -
\frac{1}{m})^2$, and
$m$ is the number of terms. Then, we need to check whether
$\lambda_L \leq -\frac{a}{b}\leq 1$ holds. If it holds, the optimal
linear coefficient $\hat \lambda$ for the optimization problem in
Eq.~\ref{eqa:minimumCorrelation2} is $-\frac{a}{b}$. Otherwise, we
just compare the values of $[\rho(\hat{l}(R,I_{\overline{S}}),\
I_S)]^2$ $w.r.t.$ $\hat{\lambda}=1$ and $\hat{\lambda}=\lambda_L$, in
order to get the optimal $\hat \lambda$ of the objective function in Eq.~\ref{eqa:minimumCorrelation2}.

\section*{Generalized Analysis of DSM}
\label{sec:GADSM}
Now, we generalize DSM's minimum correlation assumption, by extending the minimum correlation analysis to the analysis of the maximum KL-divergence, the maximum symmetrized KL-divergence and the maximum JS-divergence.


\subsection*{Effect of DSM on Maximizing KL-Divergence}
\label{sec:sec:kl}
 Recall that in Section~\ref{subsec:MCA}, Proposition~\ref{pro:minimumcorrelation} shows that after the distribution separation process, the Pearson correlation coefficient between DSM's output distribution $\hat{l}(R,I_{\overline{S}})$ and the seed irrelevance distribution $I_S$ can be minimized. Here, we further analyze the effect of DSM on the KL-divergence between $\hat{l}(R,I_{\overline{S}})$ and $I_S$.

Specifically, we propose the following Proposition~\ref{pro:maxKL}, which proves that if $\hat{\lambda}$ decreases, the KL-divergence between $\hat{l}(R,I_{\overline{S}})$ and $I_S$ will be increased monotonously.

\begin{Prop}
\label{pro:maxKL}
If $\hat{\lambda}$ $(\hat{\lambda}>0)$ decreases, the KL-divergence between $\hat{l}(R,I_{\overline{S}})$ and $I_S$ will increase.
\end{Prop}

\begin{proof}
Using the simplified notations in Table~\ref{tab:SimpleNotation},
let the KL-divergence of between $\hat{l}(R,I_{\overline{S}})$ and $I_S$ be formulated as
\begin{equation}\label{eq:DivDef}
\begin{split}
D(\hat{l}(R,I_{\overline{S}}),I_S) & = \sum_{i=1}^{m} \hat{l}(R,I_{\overline{S}})(i) \log(\frac{\hat{l}(R,I_{\overline{S}})(i)}{I_S(i)}) = \sum_{i=1}^{m} \hat{l}(i) \log(\frac{\hat l(i)}{I_S(i)})
\end{split}
\end{equation}
Now let $\xi=1/{\hat \lambda}$ as we did in the proof of Proposition~\ref{pro:minimumcorrelation} (see~\cite{ZhangSigir09}). According to Eq.~\ref{eqa:computeR}, we have $\hat{l}(R,I_{\overline{S}}) = \xi \times M +
(1-\xi)\times I_S$. It then turns out that
\begin{equation}\label{eq:est_Li}
 \hat{l}(i) = \xi \times (M(i)-I_S(i)) + I_S(i).
\end{equation}
Based on Eq.~\ref{eq:DivDef} and~\ref{eq:est_Li}, we get
\begin{equation}
\begin{split}
& ~~~~D(\hat{l}(R,I_{\overline{S}}),I_S)  = \sum_{i=1}^{m} (\xi \times (M(i)-I_S(i)) + I_S(i)) \log(\frac{\xi \times (M(i)-I_S(i)) + I_S(i)}{I_S(i)})
\end{split}
\end{equation}
Let $D(\xi) = D(\hat{l}(R,I_{\overline{S}}),I_S)$. The derivative of $D(\xi)$ can be calculated as
\begin{equation}
\begin{split}
& ~~~~D'(\xi)  ~= \sum_{i=1}^{m} [M(i)-I_S(i)+(M(i)-I_S(i))\log(\frac{\xi \times (M(i)-I_S(i)) + I_S(i)}{I_S(i)})]
\end{split}
\end{equation}
Since $\sum_{i=1}^{m}M(i) =1$ and $\sum_{i=1}^{m}I_S(i) =1$, $\sum_{i=1}^{m} [M(i)-I_S(i)]$ becomes 0. We then have
\begin{equation}\label{eq:D_xi_derivative}
\begin{split}
D'(\xi) & = \sum_{i=1}^{m} (M(i)-I_S(i))\log(\frac{\xi \times (M(i)-I_S(i)) + I_S(i)}{I_S(i)}) \\
& = \sum_{i=1}^{m} (M(i)-I_S(i))\log(\frac{\xi \times (M(i)-I_S(i))}{I_S(i)}+1)
\end{split}
\end{equation}
Let the $i^{th}$ term in the summation of Eq.~\ref{eq:D_xi_derivative} is
\begin{displaymath}
D'(\xi)(i) = (M(i)-I_S(i))\log(\frac{\xi \times (M(i)-I_S(i))}{I_S(i)}+1)
\end{displaymath}
It turns out that when $M(i)>I_S(i)$ and $M(i)<I_S(i)$, $D'(\xi)(i)$ is greater than 0. When $M(i) = I_S(i)$, $D'(\xi)(i)$ is 0. However, $M(i)$ not always equals to $I_S(i)$. Therefore, $D'(\xi) = \sum_{i=1}^{m}D'(\xi)(i)$ is greater than 0.

In conclusion, we have $D'(\xi)>0$. It means that $D(\xi)$ (i.e., $D(\hat{l}(R,I_{\overline{S}}),I_S)$) increases after $\xi$ increases. Since $\lambda = 1/\xi$,  after $\hat{\lambda}$ decreases, $D(\hat{l}(R,I_{\overline{S}}),I_S)$ will increase.
\end{proof}

According to Proposition~\ref{pro:maxKL}, if $\lambda$ reduced to its lower bound $\lambda_L$, then the corresponding KL-divergence $D(l_L(R,I_{\overline{S}}),I_S)$ will be the maximum value for all the legal $\hat \lambda$ ($\lambda_L \leq \hat \lambda < 1$ ). In this case, the output distribution of DSM will have the maximum KL-divergence with the seed irrelevance distribution.

\subsection*{Effect of DSM on Maximizing symmetrized KL-Divergence}

Having shown the effect of reducing the coefficient $\hat \lambda$ on the KL-divergence between $\hat{l}(R,I_{\overline{S}})$ and $I_S$, we now investigate the effect on the symmetrized KL-divergence between two involved distributions by proving the following proposition.

\begin{Prop}
\label{pro:maxSD}
If $\hat{\lambda}$ $(\hat{\lambda}>0)$ decreases, the symmetrized KL-divergence between $\hat{l}(R,I_{\overline{S}})$ and $I_S$ will increase.
\end{Prop}
\begin{proof}
Let the symmetrized KL-divergence of between $\hat{l}(R,I_{\overline{S}})$ and $I_S$ be denoted as
\begin{equation}\label{eq:SDivDef}
SD(\hat{l}(R,I_{\overline{S}}),I_S) = D(\hat{l}(R,I_{\overline{S}}),I_S) + D(I_S, \hat{l}(R,I_{\overline{S}}))
\end{equation}
Since we have proved in Proposition~\ref{pro:maxKL} that the increasing trend of $D(\hat{l}(R,I_{\overline{S}}),I_S)$ when $\hat \lambda$ decreases, we now only need to prove the same result for $D(I_S, \hat{l}(R,I_{\overline{S}}))$, which is computed by:
\begin{equation}\label{eq:InvDivDef}
D(I_S, \hat{l}(R,I_{\overline{S}})) = \sum_{i}I_S(i)\log(\frac{I_S(i)}{\hat l(i)}) = \sum_{i} I_S(i)\log I_S(i) - \sum_{i} I_S(i)\log \hat l(i)
\end{equation}
Now let $\xi=1/{\hat \lambda}$. According to Eq.~\ref{eqa:computeR}, we have $\hat{l}(R,I_{\overline{S}}) = \xi \times M +
(1-\xi)\times I_S$. It then turns out that
\begin{equation}\label{eq:est_Ri}
 \hat l(i) = \xi \times (M(i)-I_S(i)) + I_S(i).
\end{equation}
Based on Eq.~\ref{eq:InvDivDef} and Eq.~\ref{eq:est_Ri}, we get:
\begin{equation}
D(I_S, \hat{l}(R,I_{\overline{S}})) = \sum_{i} I_S(i)\log I_S(i) - \sum_{i} I_S(i)\log (\xi \times (M(i)-I_S(i)) + I_S(i))
\end{equation}
Let $D(\xi) = D(I_S, \hat{l}(R,I_{\overline{S}}))$. The derivative of $D(\xi)$ can be calculated as
\begin{equation}
D'(\xi) = \sum_{i} \frac{-I_S(i)(M(i)-I_S(i))}{\xi \times (M(i)-I_S(i)) + I_S(i)} = \sum_{i} \frac{-I_S(i)(M(i)-I_S(i))}{\hat l(i)}
\end{equation}
Since $M(i)$ is a linear combination of $\hat l(i)$ and $I_S(i)$, $M(i)$ is a in-between value of $\hat l(i)$ and $I_S(i)$. In other words, if $M(i)>I_S(i)$, then $\hat l(i)>M(i)>I_S(i)$, while $\hat l(i)<M(i)<I_S(i)$ if $M(i)<I_S(i)$.

If $M(i)>I_S(i)$, since $M(i)-I_S(i)>0$ and $0<\frac{I_S(i)}{\hat l(i)}<1$, we  have
\begin{displaymath}
\frac{-I_S(i)(M(i)-I_S(i))}{\hat l(i)} > -(M(i)-I_S(i))
\end{displaymath}
If $M(i)<I_S(i)$, since $M(i)-I_S(i)<0$ and $\frac{I_S(i)}{\hat l(i)}>1$, we have
\begin{displaymath}
\frac{-I_S(i)(M(i)-I_S(i))}{\hat l(i)} > -(M(i)-I_S(i))
\end{displaymath}
We then have
\begin{equation}\label{eq:SD_derivative_greaterthan0}
\begin{split}
D'(\xi) & = \sum_{i} \frac{-I_S(i)(M(i)-I_S(i))}{\hat l(i)} \\
& = \sum_{i:M(i)>I_S(i)} \frac{-I_S(i)(M(i)-I_S(i))}{\hat l(i)} + \sum_{i:M(i)<I_S(i)} \frac{-I_S(i)(M(i)-I_S(i))}{\hat l(i)} + \sum_{i:M(i)=I_S(i)} \frac{-I_S(i)(M(i)-I_S(i))}{\hat l(i)} \\
& >  \sum_{i:M(i)>I_S(i)} -(M(i)-I_S(i)) + \sum_{i:M(i)<I_S(i)} -(M(i)-I_S(i)) + \sum_{i:M(i)=I_S(i)} -(M(i)-I_S(i)) \\
&=\sum_{i} -(M(i)-I_S(i)) \\
&= 0
\end{split}
\end{equation}
We now have $D'(\xi)>0$. It means that $D(\xi)$ (i.e., $D(I_S,\hat{l}(R,I_{\overline{S}}))$) will increase after $\xi$ increases. Since $\lambda = 1/\xi$,  after $\hat{\lambda}$ decreases, $D(I_S,\hat{l}(R,I_{\overline{S}})$ will increase. Combined with the result proved in Proposition~\ref{pro:maxKL}, we can conclude that  when $\hat{\lambda}$ decreases, the symmetrized KL-divergence $D(\hat{l}(R,I_{\overline{S}}),I_S)$ + $D(I_S,\hat{l}(R,I_{\overline{S}}))$ will increase monotonically.
\end{proof}
According to Proposition~\ref{pro:maxSD}, if $\lambda$ reduced to its lower bound $\lambda_L$, then the corresponding symmetrized KL-divergence$D(I_S,\hat{l}(R,I_{\overline{S}}))$ will be the maximum value. In this case, the output distribution of DSM will have the maximum symmetrized KL-divergence with the seed irrelevance distribution.

\subsection*{Effect of DSM on Maximizing JS-Divergence}
Now, let us further study the reduction of the coefficient $\hat \lambda$ on it role in  maximizing the JS-divergence between DSM's output distribution $\hat{l}(R,I_{\overline{S}})$ and the seed irrelevance distribution $I_S$.

\begin{Prop}\label{pro:maxJS}
If $\hat{\lambda}$ $(\hat{\lambda}>0)$ decreases, the JS-divergence between $\hat{l}(R,I_{\overline{S}})$ and $I_S$ will increase.
\end{Prop}
\begin{proof}
Let the JS-divergence of between $\hat{l}(R,I_{\overline{S}})$ and $I_S$ be denoted as
\begin{equation}\label{eq:JS}
JS(\hat{l}(R,I_{\overline{S}}),I_S) = \frac{1}{2} (D(\hat{l}(R,I_{\overline{S}}),\frac{\hat{l}(R,I_{\overline{S}})+I_S}{2})+D(I_S,\frac{\hat{l}(R,I_{\overline{S}})+I_S}{2}))
\end{equation}
Now let $\xi=1/{\hat \lambda}$. Based on Eq.~\ref{eq:JS} and Eq.~\ref{eq:est_Ri}, we get
\begin{equation}
\begin{split}
JS(\hat{l}(R,I_{\overline{S}}),I_S) & = \frac{1}{2} \sum_i (\xi \times (M(i)-I_S(i)) + I_S(i)) \log(\frac{2\xi \times (M(i)-I_S(i)) + 2I_S(i)}{\xi \times (M(i)-I_S(i)) + 2I_S(i)}) \\
&~~~ + \frac{1}{2}\sum_i I_S(i) \log(\frac{2I_S(i)}{\xi \times (M(i)-I_S(i)) + 2I_S(i)})
\end{split}
\end{equation}
Let $J(\xi) = 2\times JS(\hat l,I_S)$, we can have
\begin{equation}
\begin{split}
J(\xi) & = \sum_i(\xi \times (M(i)-I_S(i)) + I_S(i))\log(2\xi \times (M(i)-I_S(i)) + 2I_S(i)) \\
&~~~~ - \sum_i(\xi \times (M(i)-I_S(i)) + I_S(i))\log(\xi \times (M(i)-I_S(i)) + 2I_S(i)) \\
&~~~~ + \sum_i I_S(i) \log 2I_S(i) - \sum_i I_S(i)\log(\xi \times (M(i)-I_S(i)) + 2I_S(i))
\end{split}
\end{equation}
The derivative of $J(\xi)$ can be calculated as
\begin{equation}
\begin{split}
J'(\xi) & = \sum_i(M(i)-I_S(i))\log(2\xi \times (M(i)-I_S(i)) + 2I_S(i))  \\
& ~~~~+ \sum_i(\xi \times (M(i)-I_S(i)) + I_S(i))\frac{M(i)-I_S(i)}{\xi \times (M(i)-I_S(i)) + I_S(i)} \\
& ~~~~- \sum_i (M(i)-I_S(i))\log(\xi \times (M(i)-I_S(i)) + 2I_S(i)) \\
& ~~~~- \sum_i(\xi \times (M(i)-I_S(i)) + I_S(i))\frac{M(i)-I_S(i)}{\xi \times (M(i)-I_S(i)) + 2I_S(i)} \\
& ~~~~ -\sum_i I_S(i)\frac{M(i)-I_S(i)}{\xi \times (M(i)-I_S(i)) + 2I_S(i)}
\end{split}
\end{equation}
Since $\hat l(i) = \xi \times (M(i)-I_S(i)) + I_S(i)$ (see Eq.~\ref{eq:est_Ri}), we have
\begin{equation}
\begin{split}
J'(\xi) & = \sum_i(M(i)-I_S(i))\log(2\hat l(i))  + \sum_i \hat l(i) \frac{M(i)-I_S(i)}{\hat l(i)} - \sum_i (M(i)-I_S(i))\log(\hat l(i) + I_S(i)) \\
& ~~~~- \sum_i \hat l(i)\frac{M(i)-I_S(i)}{\hat l(i) + I_S(i)} -\sum_i I_S(i)\frac{M(i)-I_S(i)}{\hat l(i)+I_S(i)} \\
& = \sum_i(M(i)-I_S(i))\log(\frac{2\hat l(i)}{\hat l(i)+I_S(i)}) - \sum_i  (\hat l(i)+I_S(i)) \frac{M(i)-I_S(i)}{\hat l(i)+I_S(i)} \\
& = \sum_i(M(i)-I_S(i))\log(\frac{2\hat l(i)}{\hat l(i)+I_S(i)})
\end{split}
\end{equation}
If $M(i)>I_S(i)$, since $M(i)-I_S(i)>0$ and $\hat l(i)>I_S(i)\geq 0$, we  have
\begin{displaymath}
(M(i)-I_S(i))\log(\frac{2\hat l(i)}{\hat l(i)+I_S(i)}) >0
\end{displaymath}
If $M(i)<I_S(i)$, since $M(i)-I_S(i)<0$ and $0 \leq \hat l(i) < I_S(i)$, we have
\begin{displaymath}
(M(i)-I_S(i))\log(\frac{2\hat l(i)}{\hat l(i)+I_S(i)}) >0
\end{displaymath}
We then have $J'(\xi)>0$. It means that $JS(\hat{l}(R,I_{\overline{S}}),I_S)$) increases after $\xi$ increases. Since $\lambda = 1/\xi$,  after $\hat{\lambda}$ decreases, $JS(\hat{l}(R,I_{\overline{S}}),I_S)$ will increase monotonically.
\end{proof}
According to Proposition~\ref{pro:maxJS}, if $\lambda$ reduced to its lower bound $\lambda_L$, then the corresponding JS-divergence $JS(\hat{l}(R,I_{\overline{S}}),I_S)$ will be the maximum value. In this case, the output distribution of DSM will have the maximum symmetrized JS-divergence with the seed irrelevance distribution.

\section*{Analysis of Relation between DSM and Mixture Model Feedback}
\label{sec:sec:mmf}

A related model of DSM is the Mixture Model Feedback (MMF) approach which assumes that feedback documents are generated from a mixture model with two multinomial components, i.e., the query topic model and the collection model~\cite{Zhai01model-basedfeedback}. In this section we theoretically investigate the linear combination condition of DSM in a related Mixture Model~\cite{Zhai01model-basedfeedback}.

The estimation of the output ``relevant'' query model of MMF is trying to purify the feedback document by eliminating the effect of the collection model, since the collection model contains background noise which can be regarded as the ``irrelevant'' content in the feedback document~\cite{Zhai01model-basedfeedback}. In this sense, similar to DSM, the task of MMF can also be regarded as removing the irrelevant part in the mixture model. However, to our knowledge, researchers have not investigated if the linear combination assumption is valid or not in MMF. We will theoretically prove that the mixture model in MMF is indeed a linear combination of ``relevant'' and ``irrelevant'' parts. This theoretical result can lead to a simplified version of MMF based on linear separation algorithm of DSM. Next, we first review the Mixture Model Feedback Approach in detail.

\subsection*{Review of Mixture Model Feedback Approach}
In the Mixture Model Feedback (MMF) approach, the likelihood of feedback documents ($\mathcal{F}$) can be written as:
\begin{align}
\label{eq:emlikelihood}
\log p(\mathcal{F}|\theta_F) = \sum_{d\in \mathcal{F}}\sum_{w\in d}c(w;d)\log[\lambda p(w|\theta_F) + (1-\lambda) p(w|C)]
\end{align}
where $c(w;d)$ is the count of a term $w$ in a document $d$, $p(w|\theta_F)$ is the query topic model which can be regarded as the relevance distribution to be estimated, and $p(w|C)$ is the collection distribution containing the background information. The empirically assigned parameter $\lambda$ is the amount of true relevance distribution and  $1-\lambda$ indicates the amount of background noise, i.e., the influence of $C$ in the feedback documents. An EM method \cite{Zhai01model-basedfeedback} is developed to estimate the relevance distribution via maximizing the likelihood in Equation \ref{eq:emlikelihood}. It contains iterations of two steps~\cite{zhai2007note}:
\begin{align}
\label{eq:e-step}
&p(z_w=1|\mathcal{F}, \theta_F^{(n)}) = \frac{(1 - \lambda) p(w|C)}{\lambda p(w|\theta_F^{(n)}) + (1 - \lambda) p(w|C)} &E\ step\\
\label{eq:m-step}
&p(w|R^{(n+1)}) = \frac{\sum_{d\in \mathcal{F}}(1 - p(z_w=1|\mathcal{F}, \theta_F^{(n)}))c(w,d)}{\sum_{d\in \mathcal{F}}\sum_{w^*\in V}(1 - p(z_{w^*}=1|\mathcal{F}, \theta_F^{(n)}))c(w^*,d)} &M\ step
\end{align}
where $p(z_w=1|\mathcal{F}, \theta_F^{(n)})$ is the probability that the word $w$ is from background distribution, given the current estimation of relevance distribution ($\theta_F^{(n)}$). This estimation can be regarded as a procedure to obtain relevant information from feedback documents while filtering the influence of collection distribution, leading to a more discriminative relevance model. It should be noted that in Eq.~\ref{eq:emlikelihood}, due to the $\log$ operator within the summations (i.e., $\sum_{d\in \mathcal{F}}\sum_{w\in d}c(w;d)$), it does not directly show that the mixture model is a linear combination of the collection model and the query topic model. Therefore, an EM algorithm is adopted to estimate the query topic model $\theta_F$.

\subsection*{The Simplification of EM Algorithm in MMF via Linear Separation Algorithm}
Now, we first explore the connections between DSM and MMF. Once $\lambda$ is given (either by the estimation in DSM or by an assigned value in MMF), the next step is to estimate the true relevance distribution $R$. We will demonstrate that if the EM algorithm (in MMF) converges, the mixture model of the feedback documents is a linear combination of the collection model and the output model of the EM iterative algorithm.


\begin{Prop}\label{pro:em}
If the EM algorithm (in MMF) converges, the mixture model of the feedback documents is a linear combination of the collection model and the output relevance model of the EM iterative algorithm.
\end{Prop}
\begin{proof}
When the EM method converges in the mixture model feedback approach, without loss of generality, we let $p(w|\theta_F^{(n+1)}) = p(w|\theta_F^{(n)})$. In addition to this, we can replace the $p(z_w=1|\mathcal{F}, \theta_F^{(n)})$ in Equation \ref{eq:m-step} using Equation \ref{eq:e-step}:
\begin{equation}
\label{eq:converge}
\begin{split}
p(w|\theta_F^{n+1}) = p(w|\theta_F^{(n)})
=& \frac{\sum_{d\in \mathcal{F}}\left[1 - \frac{(1 - \lambda) p(w|C)}{\lambda p(w|\theta_F^{(n)}) + (1 - \lambda) p(w|C)}\right]\cdot c(w,d)}{\sum_{d\in \mathcal{F}}\sum_{w^*\in V}\left[1 - \frac{(1 - \lambda) p(w^*|C)}{\lambda p(w^*|\theta_F^{(n)}) + (1 - \lambda) p(w^*|C)}\right]\cdot c(w^*,d)}\\
=& \frac{\sum_{d\in \mathcal{F}}\frac{\lambda p(w|\theta_F^{(n)})}{\lambda p(w|\theta_F^{(n)}) + (1 - \lambda) p(w|C)}\cdot c(w,d)}{\sum_{d\in \mathcal{F}}\sum_{w^*\in V}\frac{\lambda p(w^*|\theta_F^{(n)})}{\lambda p(w^*|\theta_F^{(n)}) + (1 - \lambda) p(w^*|C)}\cdot c(w^*,d)}
\end{split}
\end{equation}
By dividing $p(w|\theta_F^{(n)})$ in both the second term and the fourth term in Eq.~\ref{eq:converge}, we get:
\begin{align}
1 = \frac{\sum_{d\in \mathcal{F}}\frac{\lambda}{\lambda p(w|\theta_F^{(n)}) + (1 - \lambda) p(w|C)}\cdot c(w,d)}{\sum_{d\in \mathcal{F}}\sum_{w^*\in V}\frac{\lambda p(w^*|\theta_F^{(n)})}{\lambda p(w^*|\theta_F^{(n)}) + (1 - \lambda) p(w^*|C)}\cdot c(w^*,d)}
\end{align}
Then, it turns out that, for a particular word $w$:
\begin{align}
\sum_{d\in \mathcal{F}}\frac{c(w,d)}{\lambda p(w|\theta_F^{(n)}) + (1 - \lambda) p(w|C)}
= \sum_{d\in \mathcal{F}}\sum_{w^*\in V}\frac{p(w^*|\theta_F^{(n)})c(w^*,d)}{\lambda p(w^*|\theta_F^{(n)}) + (1 - \lambda) p(w^*|C)}\notag
\end{align}
Replace $\sum_{d\in \mathcal{F}}c(w,d_i)$ with $c(w^*,\mathcal{F})$, we can get:
\begin{align}
\frac{c(w,\mathcal{F})}{\lambda p(w|\theta_F^{(n)}) + (1 - \lambda) p(w|C)}
\label{emproof1}
= \sum_{w^*\in V}\frac{p(w^*|\theta_F^{(n)})c(w^*,\mathcal{F})}{\lambda p(w^*|\theta_F^{(n)}) + (1 - \lambda) p(w^*|C)}
\end{align}
If each side of Equation \ref{emproof1} is multiplied by $(1-\lambda) p(w|C)$ then it becomes:
\begin{align}
\label{subsum}
\frac{(1-\lambda) p(w|C)c(w,\mathcal{F})}{\lambda p(w|\theta_F^{(n)}) + (1 - \lambda) p(w|C)}
= (1-\lambda) p(w|C)\sum_{w^*\in V}\frac{p(w^*|\theta_F^{(n)})c(w^*,\mathcal{F})}{\lambda p(w^*|\theta_F^{(n)}) + (1 - \lambda) p(w^*|C)}
\end{align}
We can obtain the Equation \ref{subsum} for any word $w$ in the vocabulary, and now we sum them together:
\begin{align}
\label{sum}
\sum_{w^{*}\in V}\frac{(1-\lambda) p(w^{*}|C)c(w^{*},\mathcal{F})}{\lambda p(w^*|\theta_F^{(n)}) + (1 - \lambda) p(w^*|C)}
= \sum_{w^*\in V}\frac{(1-\lambda) p(w^*|\theta_F^{(n)})c(w^*,\mathcal{F})}{\lambda p(w^*|\theta_F^{(n)}) + (1 - \lambda) p(w^*|C)}
\end{align}
then we add $\sum_{w^{*}\in V}\frac{\lambda p(w^{*}|\theta_F^{(n)})c(w^{*},\mathcal{F})}{\lambda p(w^*|\theta_F^{(n)}) + (1 - \lambda) p(w^*|C)}$ to both sides of Equation \ref{sum}:
\begin{align}
\label{emproof2}
\sum_{w^{*}\in V}c(w^{*},\mathcal{F}) =& \sum_{w^*\in V}\frac{p(w^*|\theta_F^{(n)})c(w^*,\mathcal{F})}{\lambda p(w^*|\theta_F^{(n)}) + (1 - \lambda) p(w^*|C)}
\end{align}
According to Equation \ref{emproof1} and Equation \ref{emproof2}, we get:
\begin{align}
\frac{c(w,\mathcal{F})}{\lambda p(w|\theta_F^{(n)}) + (1 - \lambda) p(w|C)} = \sum_{w^{*}\in V}c(w^{*},\mathcal{F})
\end{align}
namely,
\begin{align}\label{eq:MMF_lc1_appendix}
\lambda p(w|\theta_F^{(n)}) + (1 - \lambda) p(w|C) = \frac{c(w,\mathcal{F})}{\sum_{w^{*}\in V}c(w^{*},\mathcal{F})} = tf(w,\mathcal{F})
\end{align}
where $tf(w,\mathcal{F})$ is the term frequency in the feedback documents. The above equation illustrates that the estimated distribution $\theta_F^{(n)}$ fits in a linear combination as used in Eq.~\ref{eqa:firstLC} of DSM.
\end{proof}

This proposition actually can be formulated as:

\begin{equation}\label{eq:MMF_lc1}
\lambda p(w|\theta_F^{(n)}) + (1 - \lambda) p(w|C) = \frac{c(w,\mathcal{F})}{\sum_{w^{*}\in V}c(w^{*},\mathcal{F})} = tf(w,\mathcal{F})
\end{equation}
where $tf(w,\mathcal{F})$ is the mixture model which represents the term frequency in the feedback documents, $p(w|C)$ is the collection model and $p(w|\theta_F^{(n)})$ is the estimated relevance model output by the $n^{th}$ step of the EM iterative algorithm in MMF. The above equation can be derived as:
\begin{align}\label{eq:MMF_lc2}
p(w|\theta_F^{(n)}) = \frac{1}{\lambda}\cdot tf(w,\mathcal{F}) + \left(1 - \frac{1}{\lambda}\right)\cdot p(w|C)
\end{align}
The idea is that we can regard $p(w|\theta_F^{(n)})$ as an estimated relevance distribution, $tf(w,\mathcal{F})$ as a kind of mixture distribution, and $p(w|C)$ as a kind of irrelevance distribution. Then,  Eq.~\ref{eq:MMF_lc1} fits Eq.~\ref{eqa:firstLC}, and Eq.~\ref{eq:MMF_lc2} is the same distribution separation process as the Eq.~\ref{eqa:computeR}, where $\hat{l}(R,I_{\overline{S}})$ is the estimated relevance distribution. It further demonstrates that the EM iterative steps in MMF can actually be simplified to the linear separation solution similar to Eq.~\ref{eqa:computeR}.


Another simplified solution to MMF was proposed in \cite{DBLP:journals/ipm/ZhangX08}. This solution is derived by Lagrange multiplier method, and the complexity of its divide \& conquer algorithm is $O(n)$ (on average) to $O(n^2)$ (the worst case). On the other hand, our simplified solution in Eq.~\ref{eq:MMF_lc2} was analytically derived from the convergence condition of the EM method in the MMF approach, and the complexity of the linear combination algorithm in Eq.~\ref{eq:MMF_lc2} is further reduced to a fixed linear complexity, i.e., $O(n)$.

Besides of providing a simplified solution with linear complexity to the EM method in MMF, DSM shows an essential difference regarding the coefficient $\lambda$. In MMF, the proportion of relevance model in the assumed mixture model $tf(w,\mathcal{F})$ is controlled by $\lambda$, which is a free parameter and is empirically assigned to a fixed value before running the EM algorithm. On the other hand, in DSM, as aforementioned in Section~\ref{sec:DSM}, $\lambda$ for each query is estimated via an analytical procedure based on its lower bound analysis (see Section~\ref{subsec:DMRSND}), a minimum correlation analysis (see Section~\ref{subsec:MCA}), and a maximal KL-divergence analysis described in Section~\ref{sec:sec:kl}.

\section*{Experiments}
We have theoretically described the relation between Mixture Model Feedback (MMF) approach and our DSM method. Experiments in this section provide empirical comparisons of these two methods.

\begin{figure*}
\centering
\includegraphics[height=0.42\textwidth]{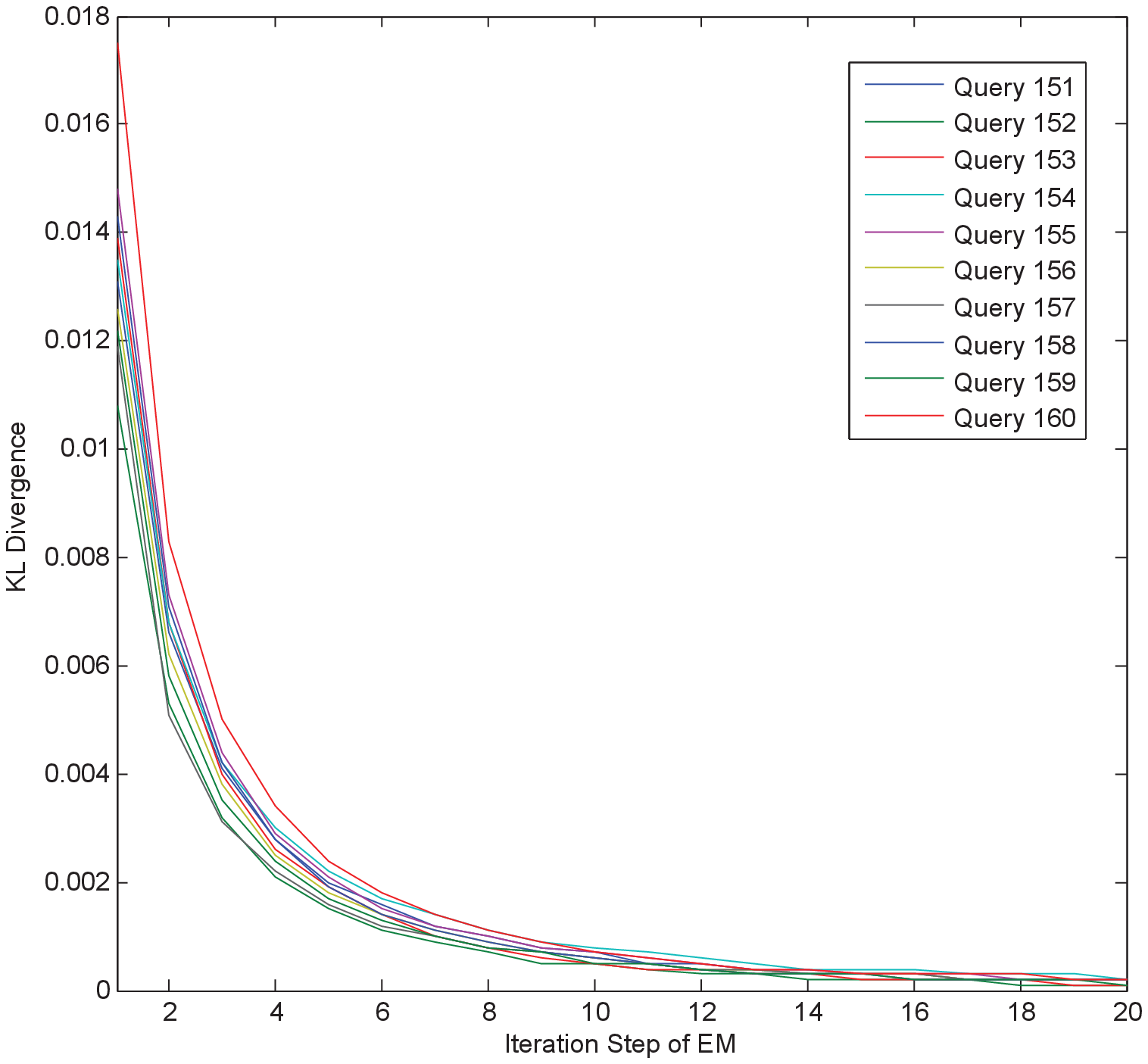}
\caption{KL Divergence between result of Equation \ref{eqa:computeR} and the current estimation of EM\label{fig:dsm-mmf}}
\end{figure*}

\begin{table*}[th]
\centering
\begin{tabular}{||l|c|c|c|c||} \hline
 MAP (chg\% over MMF) & WSJ8792 & AP8889 & ROBUST2004 & WT10G\\ \hline
 MMF & \tabincell{c}{0.3388\\($\lambda=0.2$)} & \tabincell{c}{0.3774\\($\lambda=0.2$)} & \tabincell{c}{0.2552\\($\lambda=0.1$)} & \tabincell{c}{0.1282\\($\lambda=0.3$)}\\ \hline
 DSM ($\lambda$ fixed) & \tabincell{c}{{0.3386}\\($\lambda=0.2$)} & \tabincell{c}{{0.3767}\\($\lambda=0.2$)} & \tabincell{c}{{0.2487}\\($\lambda=0.1$)}& \tabincell{c}{{0.1267}\\($\lambda=0.3$)}\\ \hline
 DSM- & 0.3474(+2.54\%)$^{*}$ & 0.3870(+2.54\%)$^{*}$ & 0.2889(+13.21\%)$^{**}$ & \textbf{0.1735}(+35.34\%)$^{**}$\\ \hline
 DSM & \textbf{0.3565}(+5.22\%)$^{**}$ & \textbf{0.3915}(+3.74\%)$^{*}$ & \textbf{0.2957}(+15.87\%)$^{**}$ & \textbf{0.1735}(+35.34\%)$^{**}$\\ \hline
 \multicolumn{5}{c}{Statistically significant improvement over MMF at level 0.05(*) and 0.01(**).}
 \end{tabular}
 \caption{Comparison of DSM and MMF approach with $TF$ and $C$ as input distributions \label{tab:MMF}}
\end{table*}

As aforementioned, the EM iteration algorithm of MMF can be simplified as a distribution separation procedure (see Equation~\ref{eq:MMF_lc2}) whose inputs are two distributions $tf(w,\mathcal{F})$ ($TF$ for short) and $p(w|C)$, where $TF$ is the distribution for which the probability of a term is its frequency in feedback documents, and $C$ is the collection distribution of terms. It has been shown that Equation~\ref{eq:MMF_lc2} is actually a special DSM when $TF$ and $C$ are DSM's input distributions and $\lambda$ is assigned empirically without principled estimation,  denoted as DSM ($\lambda$ fixed) (see Table \ref{tab:MMF}). Now, we compare MMF and DSM ($\lambda$ fixed) to test Proposition~\ref{pro:em} empirically.

At first, we directly measure the KL-divergence between resultant distributions of MMF and DSM ($\lambda$ fixed). We report the results of Query 151 - 160 on WSJ8792 with $\lambda=0.8$ in Figure \ref{fig:dsm-mmf}, and results of other queries/datasets show similar trends. It can be observed from this figure that the KL divergence between the results of two mentioned methods, i.e., MMF and DSM ($\lambda$ fixed) tends to be very close to zero, which supports the proof of their equivalence, i.e., Proposition~\ref{pro:em}.

Next, we compare the retrieval performance of MMF and DSM ($\lambda$ fixed). For MMF, we set $\lambda$ to the value with the best performance, and this optimal value is also used in DSM with $\lambda$ fixed. Experimental results are shown in Table \ref{tab:MMF}. We can find that performances of these two methods are very close, which is consistent with the analysis in Section \ref{sec:sec:mmf}. This result again confirms that the EM methods in MMF can be simplified by Equation~\ref{eq:MMF_lc2} with little decline of performance.


In addition, we also apply DSM- and DSM in the same setting (i.e., when $TF$ and $C$ are input to DSM--/DSM as the mixture distribution and seed irrelevance distribution, respectively) for comparison. It is demonstrated in Table \ref{tab:MMF} that the performances of both DSM-- and DSM are significantly better than MMF. This is because although MMF and DSM($\lambda$ fixed) empirically tune $\lambda$ for each collection, the value of this parameter is the same for each query given the test collection. On the contrary, DSM- and DSM adopt the principled estimation of $\lambda$ \textit{adaptively} for each query based on lower-bound analysis, minimum correlation analysis and maximum KL-divergence analysis. This set of experiments demonstrate that the estimation method for $\lambda$ in DSM method is crucial and effective for background noise elimination.

\section*{Conclusion}
In this paper, we further study the theoretically properties of Distribution Separation Method (DSM). Specifically, we generalized the minimum correlation analysis in DSM to maximum (original and symmetrized) KL-divergence analysis. We also proved that the solution to the well-known Mixture Model Feedback can be simplified using the linear combination technique in DSM, and this is also empirically verified using standard TREC datasets. Equipped with these analysis, DSM now has solid theoretical foundation which makes its possible to further extend DSM with principle. In addition, comparison with MMF helps us to find more scenarios for application of DSM.

\bibliography{puremodel}









\end{document}